\documentclass[11pt]{article}

\usepackage{amssymb,amsthm,amsmath,amssymb,wrapfig,dsfont,authblk}
\usepackage{thmtools,thm-restate}
\PassOptionsToPackage{hyphens}{url}
\usepackage[square,sort,semicolon,numbers]{natbib}%[round] 
\usepackage{varioref}
\usepackage[usenames,svgnames,xcdraw,table]{xcolor}
\definecolor{DarkBlue}{rgb}{0.1,0.1,0.5}
\definecolor{DarkGreen}{rgb}{0.1,0.5,0.1}
\usepackage[backref=page]{hyperref}
\hypersetup{
     colorlinks   = true,
     linkcolor    = DarkBlue, % color of internal links
     urlcolor     = DarkBlue, % color of external links
	 citecolor    = DarkGreen % color of links to bibliography
}
\renewcommand*{\backref}[1]{}
\renewcommand*{\backrefalt}[4]{%
    \ifcase #1 (Not cited.)%
    \or        (Cited on page~#2)%
    \else      (Cited on pages~#2)%
    \fi}
\usepackage[capitalise,noabbrev]{cleveref}
\usepackage[margin=0.9in]{geometry}

\usepackage[linesnumbered,lined,boxed,ruled,vlined]{algorithm2e} % For algorithms

\SetKwInOut{Parameter}{Parameter}
\SetAlFnt{\small}
\SetAlCapFnt{\small}
\SetAlCapNameFnt{\small}
\SetAlCapHSkip{0pt}
\IncMargin{-\parindent}

\usepackage{tikz}
\usetikzlibrary{fit,calc}
\colorlet{pink}{red!40}
\colorlet{blue}{cyan!60}
\colorlet{mygray}{gray!55}

\makeatletter
\renewcommand{\paragraph}{%
  \@startsection{paragraph}{4}%
  {\z@}{1.0ex \@plus 1ex \@minus .2ex}{-1em}%
  {\normalfont\normalsize\bfseries}%
}
\let\oldnl\nl% Store \nl in \oldnl
\newcommand{\nonl}{\renewcommand{\nl}{\let\nl\oldnl}}% Remove line number for one line
\makeatother

\usepackage{nicefrac}
\usepackage{enumerate}
\usepackage[labelfont={normalfont,bf},textfont=it]{caption}
\usepackage{subcaption}
\usepackage[T1]{fontenc}
\usepackage{mathtools}
\usepackage{bbm}

\usepackage{tabulary}
\usepackage{multirow}
\usepackage{booktabs}
\usepackage{colortbl}

\usepackage[backgroundcolor=blue!20!white]{todonotes}%[backgroundcolor=Moccasin,colorinlistoftodos]

\newtheorem{theorem}{Theorem}[section]

\newtheorem{lemma}[theorem]{Lemma}
\newtheorem{claim}[theorem]{Claim}
\newtheorem{corollary}[theorem]{Corollary}
\newtheorem{proposition}[theorem]{Proposition}

\newtheorem{definition}{Definition}[section]
\theoremstyle{definition}
\newtheorem{remark}[definition]{Remark}

\allowdisplaybreaks

\newcommand{\Alg}{\textrm{\textsc{Alg}}}

\newcommand{\e}{\mathbf e}
\newcommand{\s}{\mathbf s}
\newcommand{\EF}{\mathrm{EF}}
\newcommand{\EFone}{\textrm{EF1}}

\newcommand{\fPO}{\textrm{fPO}}

\newcommand{\I}{{\mathcal I}}

\newcommand{\MBB}{\mathrm{MBB}}

\newcommand{\M}{{\mathcal M}}

\newcommand{\p}{\mathbf p}

\newcommand{\PO}{\textrm{PO}}
\newcommand{\poly}{\mathrm{poly}}

\newcommand{\V}{\mathcal V}

\newcommand{\x}{\mathbf x}

\newcommand{\y}{\mathbf y}

\newcommand{\alloc}{\mathbf a}
\newcommand{\Prop}{\textrm{Prop}}
\newcommand{\Propone}{\textsc{Prop1}}
\newcommand{\EFtwo}{\textsc{EF}_1^1}
\begin{document}
\title{\bfseries On the Proximity of Markets with Integral Equilibria}

\author{Siddharth Barman\thanks{Indian Institute of Science. \texttt{barman@iisc.ac.in} \\ \hspace*{13pt} Supported by a Ramanujan Fellowship (SERB - {SB/S2/RJN-128/2015}) and a Pratiksha Trust Young Investigator Award.} \qquad Sanath Kumar Krishnamurthy\thanks{Stanford University. \texttt{sanathsk@stanford.edu}}}

\date{}
\maketitle

\begin{abstract}
We study Fisher markets that admit equilibria wherein each good is integrally assigned to some agent. While strong existence and computational guarantees are known for equilibria of Fisher markets with additive valuations~\cite{EG59consensus,orlin2010improved}, such equilibria, in general,  assign goods fractionally to agents. Hence, Fisher markets are not directly applicable in the context of indivisible goods. In this work we show that one can always bypass this hurdle and, up to a bounded change in agents' budgets, obtain markets that admit an integral equilibrium. We refer to such markets as pure markets and show that, for any given Fisher market (with additive valuations), one can efficiently compute a ``near-by,'' pure market with an accompanying integral equilibrium. %We complement this algorithmic result by proving  that it is {\rm NP}-hard to determine whether a given Fisher market admits an integral equilibrium.

Our work on pure markets leads to novel algorithmic results for fair division of indivisible goods. Prior work in discrete fair division has shown that, under additive valuations, there always exist allocations that simultaneously achieve the seemingly incompatible properties of fairness and efficiency~\cite{CKM+16unreasonable}; here fairness refers to \emph{envy-freeness up to one good} ($\EFone$) and efficiency corresponds to \emph{Pareto efficiency}. However, polynomial-time algorithms are not known for finding such allocations. Considering relaxations of proportionality and $\EFone$, respectively, as our notions of fairness, we show that fair and Pareto efficient allocations can be computed in strongly polynomial time. %In particular, we consider fairness in terms of  \emph{proportionality up to one good}~\cite{conitzer2017fair} and \emph{envy-freeness up to two goods}. These fair division results highlight the applicability of our work on pure markets.
\end{abstract}

\noindent
\section{Introduction}
Fisher markets are fundamental models of resource allocation in mathematical economics~\cite{BS00compute}. Such markets consist of  a set of divisible goods along with a set of buyers who have prespecified budgets and valuations (over all possible bundles of the goods). In this work we focus on the basic setup wherein the valuations of the buyers are additive. In an equilibrium of a Fisher market, goods are assigned prices, each buyer spends its entire budget selecting only those goods that provide maximum value per unit of money spent, and the market clears. The relevance of market equilibria (specifically from a resource-allocation perspective) is substantiated by the first welfare theorem which asserts that such equilibria are always \emph{Pareto efficient} {\cite[Chapter~16]{MWG+95microeconomic}}.

The convex program of Eisenberg and Gale provides a remarkable characterization (and, in conjunction, a proof of existence) of equilibria in Fisher markets with additive valuations: the primal and dual solutions of their convex program correspond to the equilibrium allocations and  prices, respectively~\cite{EG59consensus,NRT+07AGT}. The seminal work of Arrow and Debreu~\cite{AD54existence} further shows that equilibria exist under more general market models and convex settings; see, e.g., Mas-Colell et al.~\cite{MWG+95microeconomic}. The notable aspect of the Eisenberg-Gale characterization is that---in contrast to the encompassing result of Arrow and Debreu---it provides an efficient method for finding equilibria under additive valuations. Several algorithmic results have been developed recently for computing Fisher market equilibria and, in fact, strongly polynomial-time algorithms are known for the additive case~\cite{orlin2010improved,vegh2012strongly}. 

Along with efficiency, market equilibria provide strong fairness guarantees. A well-known result of Varian~\cite{V74equity} shows that if in a market all the agents have equal budgets, then any market equilibrium---specifically called competitive equilibrium from equal incomes (\textrm{CEEI})---leads to an \emph{envy-free} allocation. Envy freeness is a standard solution concept and it deems an (fractional) allocation of the (divisible) goods to be fair if, under it, each agent prefers its own bundle over that of any other agent~\cite{F67resource}. 

However, Fisher markets do not yield a representative model in the context of indivisible goods. Such goods correspond to discrete resources (that cannot be fractionally assigned) and naturally occur in several allocation problems, e.g., course assignment~\cite{OSB10finding} and inventory pricing~\cite{rotemberg2011fair}. A market equilibrium, in general, requires a fractional assignment of goods to agents. Hence, one cannot simply consider a market with indivisible goods and expect an equilibrium outcome wherein the goods do not have to be fractionally assigned. In other words, the desirable market properties of efficiency, fairness, and computational tractability are somewhat confined to divisible goods. 

Our work shows that one can bypass this hurdle and, up to a bounded change in budgets, always obtain markets that admit \emph{integral equilibria}. Specifically, we will consider markets that admit an equilibrium wherein each good is integrally assigned to some agent. We will refer to such Fisher markets as \emph{pure markets}. Of cou(rse, not all markets are pure.\footnote{Consider a market of a single good and two agents with equal budgets.} Nevertheless, the present paper shows that for every Fisher market (with additive valuations) there exists a ``nearby'' market which admits an integral equilibrium. Specifically, we prove that for any given market $\M$ one can construct---with a bounded change in the budgets---a  pure market $\M'$. Here, both the markets have the same set of agents, goods, and valuations, and the absolute change in any agent's budget is upper bounded by $\| \p\|_\infty$, where $\p$ is the equilibrium price vector of $\M$ (Theorem~\ref{thm:key-lemma} and Theorem \ref{thm:strong-poly-mark}).

Note that pure markets enable us to treat indivisible goods as divisible ones and apply standard (Fisher market) results, such as the first welfare theorem.  The fact that the resulting equilibrium is integral ensures that---independent of the analytic treatment---the final allocation does not require the discrete goods to be fractionally allocated, i.e., it conforms to a legitimate assignment of the given indivisible goods.  \\

\noindent 
{\bf Pure Markets for Discrete Fair Division.} Our work on pure markets leads to novel algorithmic results for discrete fair division. Specifically, we address fair division of indivisible goods among agents with additive valuations. {Note that there are no monetary transfers in this setup, i.e., unlike the market setting, here we do not have budgets or prices.}

Classical notions of fairness---e.g., envy-freeness and \emph{proportionality}\footnote{A division among $n$ agents is said to be \emph{proportionally fair} iff each agent gets a bundle of value at least $1/n$ times her value for the grand bundle of goods.}---typically address allocation of divisible goods and are not directly applicable in the discrete setting. For instance, while an envy-free and proportional allocation of divisible goods always exists~\cite{S80cut}, such an existential result does not hold when the goods are indivisible.\footnote{If a single indivisible good has to be allocated between two agents, then, under any allocation, the losing agent will be envious and will not achieve proportionality.}

To address this issue, in recent years cogent analogues of envy-freeness and proportionality have been proposed for addressing the discrete version of the fair-division problem. A well-studied solution concept in this line of work is \emph{envy-freeness up to one good}~\cite{B11combinatorial}: an (integral) allocation is said to be {envy-free up to one good} ($\EFone$) iff each agent prefers its own bundle over the bundle of any other agent up to the removal of one good. Along the lines of $\EFone$, a surrogate of proportionality---called \emph{proportionality up to one good}---has also been considered in prior work~\cite{conitzer2017fair}. In particular, an allocation is said to be proportional up to one good ($\Propone$) iff each agent receives its proportional share after the inclusion of one extra good in its bundle.\footnote{In a fair-division problem with $n$ agents, the proportional share of an agent $i$ is defined to the $1/n$ times the value that $i$ has for the entire set of goods.} 

The work of Lipton et al.~\cite{LMM+04approximately} shows that as long as the valuations of the agents are monotone an $\EFone$ allocation can be computed efficiently. This result is notably general, since it guarantees the existence of $\EFone$ allocations under arbitrary, combinatorial (monotone) valuations. Caragiannis et al.~\cite{CKM+16unreasonable} established another attractive feature of this solution concept: under additive valuations, there always exists an allocation which is both $\EFone$ and Pareto optimal ($\PO$). Though, polynomial-time algorithms are not known for finding such a fair and efficient allocation--the work of Barman et al.~\cite{barman2018finding} provides a pseudopolynomial time algorithm for this problem.  

Under additive valuations, an $\EFone$ allocations is also $\Propone$.  Hence, in the additive-valuations context, the result of Lipton et al.~\cite{LMM+04approximately} is also applicable to $\Propone$. Similarly, via the existence result of Caragiannis et al.~\cite{CKM+16unreasonable}, we get that if the agents' valuations are additive, then there exists an allocation that is both $\Propone$ and $\PO$.  

We will show that---in contrast to the known pseudopolynomial result for finding $\EFone$ and $\PO$ allocations~\cite{barman2018finding}---one can compute allocations that are $\Propone$ and $\PO$ in strongly polynomial time (Corollary~\ref{cor:fair-division-add-goods}). This result highlights the applicability of our work on pure markets.  

We also consider another, natural relaxation of $\EFone$, which we refer to as $\EFtwo$: this solution concept requires that any agent $i$ is not envious of any other agent $k$, up to the inclusion of one good in $i$'s bundle and the removal of one good from $k$'s bundle. We develop an efficient algorithm for computing allocations of indivisible goods that are simultaneously $\EFtwo$ and $\PO$ (Corollary~\ref{cor:ef2-division-add-goods}). 

It is relevant to note that the work of Barman et al.~\cite{barman2018finding} can also be considered as one that finds pure markets with limited change in budgets. However, in this sense, the result obtained in~\cite{barman2018finding} is not stronger than the one established in the present paper. That result does provide a stronger fairness guarantee ($\EFone$ and $\PO$ in pseudopolynomial time), but one can show that the algorithm developed in~\cite{barman2018finding} can lead to larger (than the ones obtained in the present paper) perturbations in the budgets; see Appendix~\ref{appendix:example} for a specific market instance in which the current algorithm outperforms (in terms of budget perturbations) the one developed in~\cite{barman2018finding}. Overall, the pure-market existence result obtained in this work is not weaker than the one obtained in \cite{barman2018finding}. Also, in contrast to that work, the present algorithm runs in strongly polynomial time and is able to address unequal budgets. \\

\noindent
{\bf Our Techniques:}
We establish the result for pure markets via a constructive proof. In particular, we develop an efficient algorithm that starts with an equilibrium of the given market and rounds its (fractional) allocation to obtain an integral one. In particular, our algorithm integrally assigns all the goods, which to begin were fractionally assigned. The algorithm does not alter the  prices of the goods. We obtain a pure market at the end by setting the new budgets to explicitly satisfy the budget-exhaustion condition with respect to the computed allocation and the unchanged prices.  While the algorithm is quite direct, the sequence in which it allocates the goods is fairly relevant.  A careful curation ensures that the new budgets are close to the given ones. Notably, in our empirical study (Section~\ref{section:empirical}), it takes less time to execute this rounding than to compute an equilibrium of the given Fisher market. 

In Section~\ref{section:fair-division} we show that the integral allocation we obtain (via rounding) satisfies notable fairness and efficiency guarantees. Given that fair-division methods are widely used in practice,\footnote{See, e.g., Spliddit~\cite{GP15spliddit}: {\url{http://www.spliddit.org/}}} efficient and easy-to-implement algorithms---such as the ones developed in this work---have a potential for direct impact.  \\

\noindent
{\bf Additional Related Work:} An interesting work of Babaioff et al.~\cite{babaioff2017competitive} considers markets wherein the indivisibility of goods is explicitly enforced. In particular, in their framework each agent selects its most preferred subset of goods, among all subsets that satisfy the budget constraint. Hence, fractional selection/allocations are ruled out in this setup. For such integral markets, existence of equilibria is not guaranteed. By contrast, we solely focus on pure/fractional markets, wherein equilibria necessarily exist. The key distinction here is that a pure market is a fractional market that happens to admit an integral equilibria. While a pure market is integral in the sense of Babaioff et al.~\cite{babaioff2017competitive}, the indivisibility of goods is not explicitly enforced in this framework. 

Babaioff et al.~\cite{babaioff2017competitive} characterize the existence of equilibria in integral markets with two agents, at most five goods, and generic budgets. On the other hand, this paper establishes that, in the space of Fisher markets, pure (and, hence, integral) markets are dense, up to bounded perturbations in the budgets.

\section{Notation and Preliminaries}
\label{sec:Preliminaries}
%\subsection{Market Terminology}

%The Fisher market is a fundamental construct of economics that is used to model resource-allocation settings~\cite{BS00compute}. Such markets consist of a set of buyers with specified budgets/endowments and a set of divisible goods. Subject to budget constraints, the agents are interested in buying goods that provide maximum utility per unit of money spent. Formally a market is 

Fisher market is a tuple $\M := \langle [n],[m],\V, \e \rangle$ wherein $[n] = \{1,2,\ldots,n\}$ denotes the set of agents, $[m]=\{1,2,\ldots,m\}$ denotes the set of goods, $\V = \{v_1,v_2,\ldots,v_n\}$ denotes the valuation profile, and $\e=(e_1,e_2,\ldots,e_n)$ denotes the budget vector. The valuation profile $\V$ specifies the cardinal preferences of each agent $i \in [n]$ over the set of goods $[m]$ via a valuation function $v_i: {[0,1]}^{m} \mapsto \mathbb{R}_{\geq 0}$. For any agent $i \in [n]$, the parameter $e_i\in\mathbb{R}_+$ represents agent $i$'s budget/endowment. 

A \emph{bundle} of {goods} is a vector $\s=(s_1,s_2,\ldots,s_m) \in[0,1]^m$ in which  $s_j$ represents the allocated quantity of the good $j$. In particular, the value that an agent $i \in [n]$ has for a bundle $\s \in [0,1]^m$ is denoted as $v_i(\s)$. A bundle $\s$ is said to be integral if under it each good is allocated integrally, i.e., for each $j \in [m]$ we have $s_j \in \{0,1\}$. Note that an integral bundle $\s$ corresponds to the subset of goods $\{j \in [m] \mid s_j=1\}$. If $\s$ is an integral bundle, we will overload notation and let $\s$ also denote the corresponding subset of goods, i.e., $\s:=\{j \in [m] \mid s_j=1\}$. 

Throughout, we will assume that agents have nonnegative and additive valuations, i.e., for each agent $i \in [n]$ and any bundle $\s$, we have $v_i(\s):=\sum_{j\in[m]} v_{i,j}s_{j}$, where $v_{i,j} \geq 0 $ denotes the value agent~$i$ has for good~$j$. \\

\noindent
\textbf{Allocation:} An allocation $\x \in {[0,1]}^{n \times m}$ refers to a collection of $n$ bundles $(\x_1,\x_2,\ldots,\x_n)$ where $\x_i=(x_{i,1},x_{i,2},\dots,x_{i,m})\in[0,1]^m$ is the bundle allocated to agent $i \in [n]$. Furthermore, in an allocation at most one unit of each good is allocated, i.e., 
for all $j \in [m]$, we have $\sum_{i \in [n]} x_{i,j} \leq 1$. In other words, an allocation corresponds to a fractional allocation of the goods among the agents. We will say that an allocation $\x$ is integral iff its constituent bundles are integral, $\x \in {\{0,1\}}^{n\times m}$. \\

\noindent 
\textbf{Market outcome and equilibrium:} For a Fisher market $\M=\langle [n], [m], \V, \e \rangle$, a \emph{market outcome} is tuple $( \x, \p )$ where $\x \in {[0,1]}^{n \times m}$ corresponds to an allocation and the price vector $\p =(p_1, p_2, \ldots, p_m)$ associates a price $p_g \in \mathbb{R}_{\geq 0}$ with each good $g \in [m]$. 

Given a price vector $\p$, write $\MBB_i$ to denote the set of goods that provide agent $i$ the maximum possible utility per unit of money spent, $\MBB_i : = \{ g \in [m] \mid v_{i,g} / p_g \geq v_{i, j}/p_j \text{ for all } j \in [m] \}$. $\MBB_i$ is called the maximum bang-per-buck set of agent $i$ (under the price vector $\p$) and, for ease of presentation, we will denote the maximum bang-per-buck ratio by $\MBB_i$ as well, i.e., $\MBB_i := \max_{j \in [m]} v_{i,j}/p_j$.

An outcome $( \x, \p )$ is said to an \emph{equilibrium} of a Fisher market $\M=\langle [n], [m], \V, \e \rangle$ iff it satisfies the following conditions: 
\begin{itemize}
	\item \emph{Market clearing:} each good $g \in [m]$ is either priced at zero, $p_g = 0$, or it is completely allocated, $\sum_{i=1}^n x_{i,g} =1$.  
	\item \emph{Budget exhaustion:} Agents spend their entire budget, i.e., for all $i\in[n]$, the following equality holds $\sum_{g\in[m]} x_{i,g}p_g=\x_i\cdot\p=e_i$.
	\item \emph{Maximum bang-per-buck allocation:}  Each agent $i \in [n]$ spends its budget only on optimal goods, i.e., if $x_{i,g} > 0$ for good $g\in [m]$, then $g \in \MBB_i$.
\end{itemize}

We will explicitly use the term \emph{integral equilibrium} to refer to a market equilibrium $( \x, \p )$ in which the allocation $\x$ is integral.

Recall that equilibria of markets (with additive valuations) correspond to optimal solutions of the Eisenberg-Gale convex program~\cite{EG59consensus,NRT+07AGT}. Furthermore, in the additive case, strongly polynomial-time algorithms exist for finding market equilibria~\cite{orlin2010improved,vegh2012strongly}. 

%Equilibria of markets (with linear utilities) correspond to optimal solutions of the classic Eisenberg-Gale convex program~\cite{EG59consensus,NRT+07AGT}.\footnote{Specifically, the primal and dual solutions capture the equilibrium allocation and prices, respectively. This result is established by essentially considering the KKT conditions of the Eisenberg-Gale convex program.} This remarkable characterization provides a proof of existence as well as an efficient method for computing Fisher market equilibria~\cite{DPS+08market}. In fact, strongly polynomial-time algorithms exist for finding market equilibria~\cite{orlin2010improved}. 

%The relevance of market equilibria---in particular, from a resource-allocation perspective---is substantiated by the first welfare theorem which asserts that such equilibria are always \emph{Pareto efficient}. 

The first welfare theorem ensures that equilibrium allocations are \emph{Pareto efficient}, i.e., satisfy a standard measure of economic efficiency. Specifically, for an instance $\langle [n], [m], \V \rangle$, an allocation $\x \in {[0,1]}^{n \times m}$ is said to be Pareto dominated by another allocation $\y \in {[0,1]}^{n \times m}$ if $v_i(\y_i) \geq v_i(\x_i)$, for each agent $i \in [n]$, and  $v_k(\y_k) > v_k(\x_k)$ for some agent $k \in [n]$. That is, compared to allocation $\x$, every agent is better off under $\y$  and at least one agent is strictly better off. An allocation is said to be \emph{Pareto efficient} or \emph{Pareto optimal} (\PO{}) if it is not Pareto dominated by any other allocation. 
%We will use the term \emph{fractionally Pareto efficient} (\fPO{}) to emphasize that fact an allocation is not Pareto dominated by any fractional allocation $\y \in {[0,1]}^{n \times m}$. 

\begin{definition}[Fractionally Pareto Efficient Allocation]
An allocation is said to be \emph{fractionally Pareto efficient} (\fPO{}) iff it is not Pareto dominated by any fractional allocation $\y \in {[0,1]}^{n \times m}$. 
\end{definition}
Note that an integral allocation $\x \in \{0,1\}^{n \times m} $ can be \fPO{}.

\begin{proposition}[First Welfare Theorem; Mas-Colell et al. ({\cite[Chapter~16]{MWG+95microeconomic}})]
	\label{prop:FirstWelfareTheorem}
	If $( \x, \p )$ is an equilibrium of a Fisher market with additive valuations, then the equilibrium allocation $\x$ is fractionally Pareto efficient ($\fPO$).
\end{proposition}

Along with efficiency, market equilibria are known to fair. In particular, if in a market all the agents have equal endowments, then any market equilibrium---specifically called competitive equilibrium from equal incomes (\textrm{CEEI})---leads to an \emph{envy-free} allocation~\cite{V74equity}. Envy freeness is a standard solution concept and it deems an allocation $\x$ to be fair if, under it, each agent prefers its own bundle over that of any other agent: $v_i(\x_i) \geq v_i(\x_k)$ for all $i, k \in [n]$~\cite{F67resource}. Hence, using Proposition~\ref{prop:FirstWelfareTheorem} and the result of Varian~\cite{V74equity}, we get that \textrm{CEEI} are both fair and efficient. 

However, as observed earlier, equilibrium allocations are not guaranteed to be integral. That is, with indivisible goods, one can not directly apply the market framework and hope to retain the desirable properties of efficiency, fairness, computational tractability, or even universal existence.  

Our work shows that interestingly, up to a bounded change in the endowments, one can always bypass this hurdle and obtain integral equilibria. Towards this end, the following notion will be useful. 

\begin{definition}[Pure Market]
	A Fisher market is said to be pure iff it admits an integral equilibrium. 
\end{definition}

As mentioned previously, pure markets enable us to treat indivisible goods as divisible ones and apply standard (Fisher market) results, such as the first welfare theorem.  The fact that the resulting equilibrium is integral ensures that---independent of the analytic treatment---the final allocation does not require the discrete goods to be fractionally allocated, i.e., it conforms to a legitimate assignment of the given indivisible goods.  \\

\noindent
{\bf Spending Graph:} We will use the construct of a spending graph to state and analyze our algorithm. Given a market $\M =\langle [n], [m], \V, \e \rangle$ along with an outcome $( \x, \p )$, the \emph{spending graph} $G(\x,\p)$ is a weighted bipartite graph whose (bipartition) vertex sets correspond to the set of agents $[n]$ and the set of goods $[m]$, respectively. In the spending graph, we have an edge $(i,j)$ between agent~$i$ and good~$j$ if and only if $x_{i,j}>0$. The weight of any edge $(i,j)$ in $G(\x, \p)$ is the amount that agent~$i$ is spending on good~$j$, i.e., weight of edge $(i,j)$ is $x_{i,j}p_j$. 

Given a Fisher market $\M$ and an equilibrium $(\x,\p)$, it is always possible to rearrange the spending so that the spending graph is a forest, i.e., we can, in strongly polynomial time, find an $\x'$ such that $(\x',\p)$ is an equilibrium of $\M$ and $G(\x',\p)$ is a forest. This fact has been used in computing market equilibrium for markets~\cite{orlin2010improved} and for approximating the \emph{Nash social welfare} objective~\cite{cole2015approximating}. For completeness, we provide a proof of this result in Appendix~\ref{app:spending-forest}.

\begin{claim}\label{claim:spending-forest}
	Given a Fisher market $\M$ and its equilibrium $(\x,\p)$,  we can find---in strongly polynomial time---an (fractional) allocation $\x'$ such that $(\x',\p)$ is also an equilibrium of $\M$ and $G(\x',\p)$ is a forest.
\end{claim}

\section{On the Proximity of Pure Markets}
The main result of this section shows that for every Fisher market there always exists a ``nearby'' market which is pure. Our proof of this result is constructive. In particular, we develop a strongly  polynomial-time algorithm (\Alg) that, for any given market $\M=\langle [n], [m], \V, \e \rangle$ and its equilibrium $(\x, \p)$, finds a pure market $\M'=\langle [n], [m], \V, \e' \rangle$ such that the absolute perturbation in endowments is at most $\| \p\|_\infty$, i.e., $\| \e - \e' \|_\infty \leq \| \p \|_\infty$. $\Alg$ also computes an integral equilibrium $(\x', \p)$ of $\M'$.%, thereby certifying that it is indeed pure. 

%Note finding such an equilibrium is essential, since---in contrast to computing an arbitrary equilibrium---finding an integral equilibrium is computationally hard, i.e., determining if a given market is pure is an {\rm NP}-hard problem (Theorem~\ref{theorem:pure-hardness}). The application of this result to fair division also requires such an integral equilibrium.  

%Given a set of agents $[n]$, a set of items $[m]$, an additive valuation profile $\V=\{v_1,v_2,\dots,v_n\}$, a budget vector $\e=(e_1,e_2,\dots,e_n)$, and a market outcome $(\x,\p)$ for the corresponding fractional market. In strongly polynomial time, we will find a near-by budget vector $\e'$ such that there exists a market outcome for the fractionally integral market with this near-by budget vector $\e'$. In fact, in strongly polynomial time, we also compute an integral allocation $\x'$ such that $(\x',\p)$ is a market outcome of this fractionally integral market.

\begin{theorem}[Main Result]\label{thm:key-lemma}
	Given a Fisher market $\M=\langle [n],[m],\V, \e \rangle$ with additive valuations and its equilibrium $(\x,\p)$, we can find---in strongly polynomial time---a budget vector $\e'$ and an integral allocation $\x'$ such that
	\begin{itemize}
		\item $(\x',\p)$ is an integral equilibrium of the market $\M'=\langle [n],[m],\V, \e' \rangle$.
		\item The budget vector $\e'$ is close to $\e$: $\|\e'-\e\|_{\infty}\leq \|\p\|_{\infty}$. In addition, $\sum_{i=1}^n e'_i = \sum_{i=1}^n e_i$.
	\end{itemize}
\end{theorem} 

Note that (in contrast to computing an arbitrary equilibrium) finding an integral equilibrium is computationally hard, i.e., determining whether a given Fisher market is pure is an {\rm NP}-hard problem (Appendix~\ref{appendix:hardness}). Hence, a notable aspect of $\Alg$ is that it, along with a pure market, finds an accompanying integral equilibrium. 

\subsection{Rounding Algorithm}

Recall that, for any given market $\M$ and its equilibrium $(\x, \p)$, we can assume, without loss of generality, that the spending graph $G(\x,\p)$ is a forest. Our algorithm, $\Alg$, constructs a new (integral) allocation $\x'$ by iteratively assigning goods to agents until all the goods are allocated. In $\Alg$, we initialize $G$ to be the spending forest $G(\x,\p)$ and root each tree in $G$ at some agent. Then, we assign child goods to agents $i \in [n]$ with no parents (i.e., to root agents), until adding any more child good to $i$ would violate $i$'s original endowment (i.e., budget constraint) $e_i$. The remaining child goods are then appropriately assigned to grandchildren agents. After each such distribution, we delete this parent agent $i$ and all of its child goods (that have now been alloted). Overall, we repeat this specific method of distributing goods until $G$ is empty. 

{
 % \SetAlFnt{\footnotesize}
%\footnotesize
\begin{algorithm}[ht!]
    \SetKwInOut{Input}{Input}
    \SetKwInOut{Output}{Output}
    \SetAlgoNoLine
    \renewcommand{\thealgocf}{}
    \DontPrintSemicolon
   % \SetAlgorithmName{$\Alg$}{ }{ }
   	\caption{\Alg}
	\label{alg:goods}
	
    \Input{A Fisher market $\M = \langle [n],[m],\V,\e\rangle$ with additive valuations and an equilibrium $(\x,\p)$ of $\M$.} 
    \Output{An integral allocation $\x'$ and a budget vector $\e'$ such that $(\x',\p)$ is an integral equilibrium of the  market $\M'=\langle [n],[m],\V,\e'\rangle$ and $\|\e'-\e\|_{\infty}\leq\|\p\|_{\infty}$}
    
    Set $\x'\leftarrow (\emptyset,\emptyset,\dots,\emptyset)$, i.e., for any agent~$i$ we initialize $\x'_i\leftarrow\emptyset$ \;
	\tcc{We construct $\x'$ by assigning goods to agents until all goods are allocated}
	Initialize $G$ to be the spending forest of $(\x,\p)$, i.e., $G\leftarrow G(\x,\p)$ \;
	\tcc{Whenever we allocate a good, we delete the corresponding vertex from $G$.}
	Root each tree in the forest $G$ at some agent \label{step:root} \;
	Allocate all leaf goods to parent agents \label{step:leaf} \;
	\tcc{That is, for all $j\in[m]$ if $x_{i,j}=1$ then $x'_i\leftarrow x'_i\cup\{j\}$ and delete $j$ from $G$.}
	\While{there is an agent~$i$ with no parent (i.e., $i$ is a root node) in $G$}{
		\While{there is a good~$g$ in the neighborhood of $i$ (i.e., edge $(i,g)$ is in $G$) such that $\p(\x'_i\cup\{g\})\leq e_i$}{
			Allocate $g$ to agent $i$: update $\x'_i\leftarrow \x'_i\cup\{g\}$ and delete $g$ from $G$.
		}
		Allocate every remaining child~$j$ of $i$ to any (agent) child~$k$ of $j$ and delete $j$ from $G$. Here, $i$ and $k$ are agents and $j$ is a good \label{step:inherit}\;
		\tcc{That is, before agent~$i$'s deletion, its grandchildren inherit the remaining child goods of $i$}
		Delete agent~$i$ from $G$.\;
	}
	$\e'\leftarrow (\p(\x'_1),\p(\x'_2),\dots,\p(\x'_n))$
    
  \end{algorithm}
 }
%\textbf{Algorithm outline:} 

The integral allocation $\x'$ we construct is a rounding of the allocation $\x$. In particular, if a good is integrally allocated to agent $i$ under $\x$, then it will continue to be assigned to $i$ in $\x'$. Hence, the focus here is to analyze the assignment of goods which are fractionally allocated (i.e., are not integrally allocated) in $\x$. We will use the term \emph{contested goods} to refer to goods that are fractionally allocated in $\x$. Note that all the goods considered in  the nested while-loops of $\Alg$ are contested. 

%\footnote{Note that a good $j$ can be contested among all the $n$ agents, i.e., we can have $x_{i,j} >0$ for all $i \in [n]$; this will happen if the forest  $G(\x, \p)$ is a star graph.}

\subsection{Proof for Theorem~\ref{thm:key-lemma}}

The runtime analysis of $\Alg$ is direct and leads to following proposition. 

\begin{proposition}\label{claim:polytime}
	$\Alg$ runs in strongly polynomial time. 
\end{proposition}
%\begin{proof}
%	The algorithm $\Alg$ initializes $G$ to be a forest with $n+m$ nodes. Finding an agent with no parent in $G$ requires $\O(n+m)$ time. Furthermore, in every iteration of the inner and outer while-loop, we delete at least one vertex of $G$. Since the algorithm terminates when $G$ is empty, we get that the total number of iterations is at most $n+m$ and, hence, the time complexity of $\Alg$ is polynomial. 
%\end{proof}

In Lemma~\ref{lem:market-outcome} we will show that the output of $\Alg$, i.e., $(\x', \p)$, is an equilibrium of market $\M'=\langle [n], [m], \V, \e' \rangle $. Lemma~\ref{lem:near-by-budget} asserts that the computed endowments $\e'$ are close to given budgets $\e$. Together, Lemma~\ref{lem:market-outcome} and Lemma~\ref{lem:near-by-budget} directly imply Theorem~\ref{thm:key-lemma}. 

The following supporting claim shows that $\Alg$ maintains a useful invariant. 
\begin{claim}
	\label{claim:invariant}
	Throughout the execution of $\Alg$, the graph $G$ is a forest. In addition, the root and leaves of every tree in $G$ correspond to agents (i.e., are agent nodes).
\end{claim}
\begin{proof}
	The graph $G$ is initialized to be the spending forest, and throughout $\Alg$ we only delete vertices from $G$, without ever adding an edge. Hence, $G$ continues to be a forest.
	
	To establish the property about leaf nodes in $G$, note that in Step~\ref{step:leaf} we assign all the leaves which correspond to goods. Therefore, before the while-loop begins, all leaf nodes correspond to agents. If, for contradiction, we assume that a node $j \in [m]$---which corresponds to a good---becomes a leaf at some point of time, then this must have happened due to the deletion of $j$'s child node $i\in [n]$ (which corresponds to an agent). However, we delete an agent node $i$ only if it has no parent in $G$ (this is exactly the case in which $i$ is considered in the outer while-loop). This contradicts that fact that $\Alg$ would have deleted $i$, implying that a node $j$ (which corresponds to a good) never becomes a leaf in $G$.
	
	Finally, note that at the beginning of $\Alg$ the root nodes correspond to agents: in Step~\ref{step:root} we explicitly root the trees of $G$ at agent nodes. As before, if we assume, for contradiction, that a good node $j \in [m]$ becomes a root at some point of time, then this must have happened due to the deletion of $j$'s parent node $i\in [n]$ (which corresponds to an agent). However, we delete an agent node $i$ only after all of $i$'s child nodes (which includes $j$) have been assigned (see Step~\ref{step:inherit}). Therefore, before $i$'s deletion we would have assigned $j$ to a grandchild of $i$ (who is guaranteed to exist, due to the fact that $j$ is not a leaf node). That is, $\Alg$ would have deleted $j$ (from $G$) before $i$, contradicting the assumption that $j$ ends up being a root node. Hence, the stated claim follows for the root nodes as well.
\end{proof}

\begin{lemma}\label{lem:market-outcome}
For a given market $\M=\langle [n], [m], \V, \e \rangle$ (with additive valuations) and equilibrium $(\x, \p)$, let $\x'$ and $\e'$, respectively, be the allocation and the endowment vector computed by $\Alg$. Then,  $(\x',\p)$ is an integral equilibrium of the market $\M'=\langle [n],[m],\V,\e' \rangle$.
	% and $\| \e' - \e \|_\infty \leq \| \p \|_\infty $.
\end{lemma}
\begin{proof}
	We will first show that $\Alg$ ends up allocating every good. For any good $j \in [m]$, consider the iteration in which its parent node $i \in [n]$ is being considered in the outer while-loop, i.e., the loop after which $i$ gets deleted. Note that the parent node $i$ is guaranteed to exist since $j$ is never a root (Claim~\ref{claim:invariant}). Furthermore, the algorithm does not terminate till it deletes all the agent nodes from $G$, hence there necessarily exists a point of time when the agent node $i$ is under consideration.  
	
	By construction, good $j$ either gets assigned to $i$ or to a grandchild $k \in [n]$ of node $i$; Claim~\ref{claim:invariant} ensures that $k$ exists.  Hence, we get that all goods are allocated/deleted from $G$ over the course of the algorithm. Hence, the integral allocation $\x'$ satisfies the market clearing condition.  
	
	By construction, the allocation $\x'$, returned by $\Alg$, is a rounding of the allocation $\x$. In particular, for every agent $i \in [n]$, the set of goods that $i$ spends on in $\x'$ is a subset of the goods that $i$ spends on in $\x$, i.e., $\x'_i \subseteq \{ j \in [m] \mid x_{i,j} >0 \}$. Therefore, analogous to $\x$, in $\x'$ agents spend only on maximum bang-per-buck goods, $\x'_i \subseteq \MBB_i$; note that the prices of the goods remain unchanged. Moreover, the budget vector $\e'$ is chosen to satisfy the budget exhaustion condition. Hence $(\x',\p)$ is an integral equilibrium of the market $\M'$.
\end{proof} 

\begin{lemma}\label{lem:near-by-budget}
For any given market $\M=\langle [n], [m], \V, \e \rangle$ (with additive valuations) and equilibrium $(\x, \p)$, the budget vector $\e'$ computed by $\Alg$ satisfies $\|\e'-\e\|_{\infty}\leq \|\p\|_{\infty}$ and $\sum_{i=1}^n e'_i = \sum_{i=1}^n e_i$.
\end{lemma}
\begin{proof}
	
	In the while-loops of $\Alg$ an agent can receive only contested goods: either the parent good and/or its child goods. Agents that have no children in $G$ (at the beginning of the while loops) or are isolated satisfy the endowment bound directly; such an agent $i$ has at most one contested good, its parent $\widehat{g}$, and we have $ e_i - p_{\widehat{g}} \leq e'_i \leq e_i + p_{\widehat{g}}$. Recall that $\p(\x_i) = e_i$. Hence, to complete the proof we now need to obtain the endowment bounds for agents that have child nodes.
	
	Note that the child nodes (goods) of an agent $i$ are never deleted before $i$. The child goods are allocated/deleted only when agent $i$ is selected in the outer while-loop. If an agent $i$ has children, but it does not receive any of its child nodes, then it must be the case that $i$'s endowment is high enough to not accommodate any child, $g$. Specifically, we have $\p(\x'_i) + p_g > e_i$, i.e., $e'_i \geq e_i - p_g$. Furthermore, in this case, the only good that $i$ may have received during the execution of the while-loops is its parent good, $\widehat{g}$, hence $e'_i \leq e_i + p_{\widehat{g}}$. 
	
	The remainder of the analysis addresses agents who have children and receive at least one of their child nodes (goods). For such agents, the condition of the inner while-loop ensures that we never over allocate child nodes, $e'_i = \p(\x'_i) \leq e_i$. We will establish a lower bound for $e'_i$s by considering different cases based on whether an agent $i \in [n]$ receives all of its child nodes or just some of them. Here, we write $\widehat{g} \in [m]$ to denote the parent good of agent $i$ in $G$. 
	
	\begin{itemize}
		\item If an agent $i$ receives all of its child nodes, then $e'_i = \p(\x'_i) \geq \p(\x_i) - p_{\widehat{g}}$; here, the subtracted term, $p_{\widehat{g}}$, accounts for the fact that $i$ might not have received its parent good $\widehat{g}$. Hence, in this case we have $e'_i \geq e_i - p_{\widehat{g}}$. 
		\item In case agent $i$ does not receive a child good $g$, from the condition in the inner while-loop, we get $\p(\x'_i) + p_g > e_i$. Otherwise, child $g$ would have been included in $\x'_i$. Therefore, $e'_i = \p(\x'_i) \geq e_i - p_g$ and we get a lower bound in this case as well.   
	\end{itemize}
	
	Overall, the endowments satisfy $\| \e' - \e\|_\infty \leq \| \p \|_\infty$. 
	
	Note that $\Alg$ does not modify the prices of the goods. Since both the markets $\M$ and $\M'$ have the same equilibrium prices $\p$, the budget-exhaustion and market-clearing conditions of $\M$ and $\M'$ give us: $\sum_i e'_i = \sum_j p_j = \sum_i e_i$.
\end{proof}

\begin{remark}
	\label{remark:mbbg}
	The proof of Lemma~\ref{lem:near-by-budget} shows that if $e'_i<e_i$ then there exists a good $g\notin\x'_i$ that was fractionally allocated to $i$ under $\x$ (i.e., $x_{i,g}>0$) such that $e_i \leq e'_i + p_g$. Note that for such a good $g$ (via the maximum bang-per-buck condition in the definition of an equilibrium) we have $g \in \MBB_i$. 
	
	The analysis also ensures that if $e'_i > e_i$, then there exists a good $\widehat{g} \in \x'_i \subseteq \MBB_i$ (specifically, the parent of $i$) such that $e'_i \leq e_i + p_{\widehat{g}}$.
\end{remark}

From Proposition~\ref{claim:polytime}, Lemma~\ref{lem:market-outcome}, and Lemma~\ref{lem:near-by-budget}, we directly obtain Theorem~\ref{thm:key-lemma}.

\subsection{An Extension of Theorem~\ref{thm:key-lemma}}

The fact that Theorem~\ref{thm:key-lemma} requires an equilibrium of the given market is not a computational hurdle. The work of Orlin~\cite{orlin2010improved} provides a strongly polynomial-time algorithm for computing an equilibrium $(\x, \p)$ of a given Fisher market $\M$. Hence, Theorem~\ref{thm:key-lemma}, along with the result of Orlin~\cite{orlin2010improved}, leads to the following algorithmic result. 

\begin{theorem}\label{thm:strong-poly-mark}
	Given $m$ goods, $n$ agents with additive valuations, $\V=\{v_1, \ldots, v_n\}$, and a budget vector $\e$. In (strongly) polynomial time, we can find a budget vector $\e'$, an integral allocation $\x'$, and a price vector $\p$ such that:
	\begin{itemize}
		\item $(\x',\p)$ is an integral equilibrium of the (pure) market $\M'=\langle [n],[m],\V,\e' \rangle$.
		\item The budget vector $\e'$ is close to $\e$: $\|\e'-\e\|_{\infty}\leq \|\p\|_{\infty}$ and  $\sum_{i=1}^n e'_i = \sum_{i=1}^n e_i$.
	\end{itemize}
\end{theorem}

\section{Pure Markets for Discrete Fair Division}
\label{section:fair-division}

The section addresses the problem of fairly dividing $m$ indivisible goods among a set of $n$ agents with nonnegative, additive valuations $\V =\{ v_1,v_2,\ldots,v_n\}$. We will denote an instance of a fair division problem as a tuple $\I = \langle [n],[m],\V \rangle$.\footnote{We do not have budgets or prices in the fair division setup.} Note that for each agent $i \in [n]$ the valuation for a subset of goods $S \subseteq [m]$ satisfies $v_i(S) = \sum_{j \in S} v_{i,j} $, where $v_{i,j} \in \mathbb{R}_+$ is the value that agent $i$ has for good $j$. 

%As mentioned previously, classic notions of fairness (in particular, envy freeness and proportionality) are not representative in the context of indivisible goods. To address this issue, in recent years compelling analogues of these notions have been proposed for addressing allocation of discrete goods. 

A prominent solution concept in discrete fair division is \emph{envy-freeness up to one good}. Formally, for a fair-division instance $\I= \langle [n],[m],\V \rangle$, an integral allocation $\x = (\x_1, \x_2, \ldots, \x_n) \in \{0,1\}^{n \times m}$ is said to be envy-free up to one good ($\EFone$) iff for every pair of agents $i,k\in [n]$ there exists a good $g\in \x_k$ such that $v_i(\x_i)\geq v_i(\x_k\setminus\{g\})$. 

Strong existential guarantees are known for $\EFone$, even under combinatorial valuations: it is show in~\cite{LMM+04approximately} that as long as the valuations of the agents are monotone an $\EFone$ allocation exists and can be computed efficiently. Caragiannis et al.~\cite{CKM+16unreasonable} prove that, in the case of additive valuations, this notion of fairness is compatible with (Pareto) efficiency, i.e., there exists an allocation which is both $\EFone$ and Pareto optimal ($\PO$). However, polynomial-time algorithms are not known for finding such allocations--prior work~\cite{barman2018finding} provides a pseudopolynomial time algorithm for this problem.  

Along the lines of $\EFone$, a surrogate of proportionality---called \emph{proportionality up to one good}---has also been considered previously~\cite{conitzer2017fair}. Formally, an allocation $\x=(\x_1, \x_2, \ldots, \x_n)$ is said to be proportional up to one good ($\Propone$) iff for every agent $i \in [n]$ there exists a good $g \in [m]$ such that $ v_i(\x_i\cup \{g\}) \geq v_i([m])/n$. Write $\Prop_i$ to denote the proportional share of agent $i$, i.e., $\Prop_i := v_i([m])/n$. 

Under additive valuations, $\EFone$ allocations are also $\Propone$.  Hence, the result of Lipton et al.~\cite{LMM+04approximately} implies that $\Propone$ allocations exist when the valuations are additive. Similarly, via~\cite{CKM+16unreasonable}, we get that if the agents' valuations are additive, then there exists an allocation that is both $\Propone$ and $\PO$.  

We will show that---in contrast to the known pseudopolynomial results for finding $\EFone$ and $\fPO$ allocations~\cite{barman2018finding}---one can compute allocations that are $\Propone$ and $\fPO$ in strongly polynomial time (Corollary~\ref{cor:fair-division-add-goods}).\footnote{Recall that $\fPO$ is a stronger solution concept that $\PO$, since it requires that an allocation is not Pareto dominated by any fraction (and, hence, any integral) allocation. On the other hand, $\PO$ rules out domination solely by integral allocations.} Finding a $\Propone$ and $\PO$ allocation in polynomial time was identified as an open question in~\cite{conitzer2017fair}, and our algorithmic result for this problem highlights the applicability of Theorem~\ref{thm:strong-poly-mark} in the context of fair division of indivisible goods. 

In addition, we prove a similar result for a natural relaxation of $\EFone$, which we call envy-free up to addition of a good in the first bundle and removal of another good from the other bundle ($\EFtwo$). Formally, an integral allocation $\x = (\x_1, \x_2, \ldots, \x_n)$ is said to be $\EFtwo$ iff for every pair of agents $i, k \in [n]$, there exist goods $g_1 \in [m]$ and $g_2 \in \x_k$, such that $v_i(\x_i \cup\{g_1\})\geq v_i(\x_k\setminus\{g_2\})$. Corollary~\ref{cor:ef2-division-add-goods} shows that  an integral allocation, which is both $\EFtwo$ and $\fPO$, can be computed efficiently. 

%This solution concept requires that an agent $i$ is not envious of any other agent $k$, up to the inclusion of one good in $i$'s bundle and the removal of one good from $k$'s bundle.

\begin{corollary}
	\label{cor:fair-division-add-goods}
	Given a fair-division instance with indivisible goods and additive valuations, in strongly polynomial time we can compute an integral allocation $\alloc$ which is both $\Propone$ (fair) and $\fPO$ (efficient). 
\end{corollary}
\begin{proof}
	Given a fair-division instance $\I=\langle [n],[m],\V \rangle$, we construct a Fisher market $\M=\langle [n],[m],\V, \e=\vec{1} \rangle$ by setting the endowment of each agent equal to one. Theorem~\ref{thm:strong-poly-mark} shows that in strongly polynomial time we can compute an equilibrium $(\x,\p)$ of the market $\M$ and, then, round $\x$ to an integral allocation $\alloc$ and obtain a budget vector $\e'$ such that 
	%\begin{itemize}
	$(\alloc,\p)$ is a integral equilibrium of the market $\M'=\langle [n],[m],\V,\e' \rangle$ and the budget vector $\e'$ is close to $\e=\vec{1}$; in particular, $\|\e'-\vec{1}\|_{\infty}\leq \|\p\|_{\infty}$. % and $\|\e'\|_1=\|\vec{1}\|_1=n$. <---- WHERE DO WE USE THIS?
	%\end{itemize}
	
	Since $\alloc$ is an equilibrium of the Fisher market $\M'$, via the first welfare theorem (Proposition~\ref{prop:FirstWelfareTheorem}), we know that $\alloc$ is $\fPO$. Next, we will prove that $\alloc$ is $\Propone$ as well.

	The conditions that define an equilibrium ensure that for all agents $i \in [n]$ and goods $g \in \alloc_i$ (i.e., the goods that are allocated to $i$ in $\alloc$) we have $\frac{v_{i,g}}{p_g} = \MBB_i :=\max_{j' \in[m]} \frac{v_{i,j'}}{p_{j'}}$.\footnote{Note that the prices of the goods are the same under the two equilibria $(\x, \p)$ and $(\alloc, \p)$.} The proof of Lemma~\ref{lem:near-by-budget} further provides the guarantee that if $e'_i < e_i$, then there exists a good $g \in \MBB_i$ such that $e'_i \geq e_i - p_g$ (Remark~\ref{remark:mbbg}). Using these facts we will perform a case analysis to show that allocation $\alloc$ satisfies the stated fairness guarantee: 
	
	\begin{itemize}
		\item If $\p(\alloc_i) = e'_i < e_i = 1$, then there exists a good $g \in \MBB_i$ such that $\p(\alloc_i \cup\{g\})\geq 1$. Therefore,   
		\begin{align*}
		v_i(\alloc_i \cup\{g\}) & = \MBB_i  \  \p(\alloc_i \cup\{g\})   \tag{$v_i$ is additive and $\alloc_i \subseteq \MBB_i$} \\
		& \geq \MBB_i \cdot 1  \tag{$\p(\alloc_i \cup\{g\})\geq 1$} \\
		& = \MBB_i \cdot \p([m])/n \tag{$\p([m]) = \sum_i e_i = n$} \\
		& \geq v_i([m])/n \tag{$\MBB_i \ p_j \geq v_{i,j}$ for all goods $j$} \\
		& = \Prop_i 
		\end{align*}
		\item If $\p(\alloc_i) = e'_i \geq e_i = 1$, then
		\begin{align*}
		v_i(\alloc_i) & = \MBB_i \ \p(\alloc_i) \tag{$\alloc_i \subseteq \MBB_i$} \\
		& \geq \MBB_i \cdot 1 \\ &  = \MBB_i \ \p([m])/n \\
		& \geq v_i([m])/n \\ & = \Prop_i.
		\end{align*}
	\end{itemize}
	
	Overall, we get that for any fair-division instance $\I$, a $\Propone$ and $\fPO$ allocation can be computed in strongly polynomial time. 
\end{proof}

Next, we provide a strongly polynomial-time algorithm for finding integral allocations that are simultaneously $\EFtwo$ and $\fPO$.   

\begin{corollary}\label{cor:ef2-division-add-goods}
	Given a fair-division instance with indivisible goods and additive valuations, in strongly polynomial time we can compute an integral allocation $\alloc$ which is both $\EFtwo$ and $\fPO$. 
\end{corollary}
\begin{proof}
	Given a fair-division instance $\I=\langle [n],[m],\V \rangle$, we construct a Fisher market $\M=\langle [n],[m],\V,\e=\vec{1} \rangle$ by setting the endowment of each agent equal to one. Theorem~\ref{thm:strong-poly-mark} shows that in strongly polynomial time we can compute an equilibrium $(\x,\p)$ of the market $\M$ and, then, round $\x$ to an integral allocation $\alloc$ and obtain a budget vector $\e'$ such that 
	%\begin{itemize}
	$(\alloc,\p)$ is a integral equilibrium of the market $\M'=\langle [n],[m],\V,\e' \rangle$ and the budget vector $\e'$ is close to $\e=\vec{1}$; in particular, $\|\e'-\vec{1}\|_{\infty}\leq \|\p\|_{\infty}$. % and $\|\e'\|_1=\|\vec{1}\|_1=n$. <---- WHERE DO WE USE THIS?
	%\end{itemize}
	
	As noted in Remark~\ref{remark:mbbg}, in this construction, for each agent $i \in [n]$ we have $| e'_i - e_i | \leq p_g$ where $g$ is in fact a good that is fractionally allocated to $i$ under $\x$, i.e., $x_{i,g} >0$. Therefore, the following two properties hold 
	\begin{itemize}
		\item[{\rm P1}:] For each agent $i \in [n]$, there exists a good $g_1 \in \MBB_1$ such that $\p(\alloc_i \cup \{ g_1 \}) \geq 1$. 
		
		If $\p(\alloc_i) = e'_i < 1$,\footnote{By construction, $e_i = 1$.} then this inequality follows from the first part of Remark~\ref{remark:mbbg}. Otherwise, if  $\p(\alloc_i) = e'_i \geq 1$, then 
		the inequity holds trivially--the prices are nonnegative. 
		\item[{\rm P2}:] For each agent $k \in [n]$, there exists a good $g_2 \in \alloc_k \subseteq \MBB_k$ such that $\p(\alloc_k \setminus \{g_2\}) \leq 1$. 
		
		If $\p(\alloc_k) = e'_k > 1$, then (as stated in the second part of Remark~\ref{remark:mbbg}) we have a good $g_2 \in \alloc_k \subseteq \MBB_k$ such that $\p(\alloc_k \setminus \{g_2\}) \leq 1$. For the complementary case, $\p(\alloc_k)=e'_k \leq 1$, this inequality directly holds. 
	\end{itemize}
	
	Properties {\rm P1} and {\rm P2} imply that allocation $\alloc$ is $\EFtwo$ (here, for any two agents $i$ and $k$ we select goods $g_1$ and $g_2$ as specified in the two properties, respectively): 
	\begin{align*}
	v_i(\alloc_i \cup \{g_1\}) & = \MBB_i \ \p(\alloc_i \cup \{g_1 \}) \tag{$\alloc_i \subseteq \MBB_i$ and $g_1 \in \MBB_i$} \\
	& \geq \MBB_i \cdot 1 \tag{{\rm P1}} \\
	& \geq \MBB_i \ \p( \alloc_k \setminus \{ g_2\}) \tag{{\rm P2}} \\
	& \geq v_i (\alloc_k \setminus \{ g_2\}) \tag{$\MBB_i \ p_j \geq v_{i,j}$ for all goods $j$} 
	\end{align*}

\end{proof}

\section{Some Empirical Results}
\label{section:empirical}
For an experimental analysis of $\Alg$, we generate random instances of Fisher markets with equal incomes ($\e=\vec{1}$) and number of agents  $n\in\{2,4,8,16,32,64\}$. For each $n$, the number of goods are kept to be five times the number of agents ($m= 5 n$) and we run the experiment $100$ times. Agents' valuations for the goods are selected uniformly at random from the set $S=\{2^{2^{k-1}} \mid k\in[10]\}$ (i.e., for any agent $i\in[n]$ and any good $j\in[m]$ we pick $v_{i,j}$ uniformly at random from the set $S$). Generating the valuations this way helps avoid convergence issues while solving the Eisenberg-Gale convex program.

%In this section, we report on some experimental work done to test the performance of $\Alg$. We run our experiments on random instances of Fisher markets with equal incomes (i.e., $\e=\vec{1}$), while keeping number of goods to be five times the number of agents (i.e., $m=5n$). 

Given a Fisher market $\M=\langle [n],[5n],\V,\vec{1} \rangle$, we compute its equilibrium allocation ($\x$) using projected gradient ascent on the corresponding Eisenberg-Gale convex program.\footnote{Recall that the optimal solutions of the Eisenberg-Gale convex program correspond to   equilibrium allocations of the underlying Fisher market~\cite{EG59consensus}.} In addition, we find an equilibrium price vector ($\p$) using the  equilibrium conditions. Then, we update $\x$ using Algorithm~\ref{alg:rearrange-spending} (Appendix~\ref{app:spending-forest}) to ensure that that its spending graph is a forest and, finally, execute $\Alg$  on  the input $(\M,(\x,\p))$. Note that while there are sophisticated algorithms to compute exact equilibrium of Fisher markets in strongly polynomial time~\cite{orlin2010improved, vegh2012strongly}, we use the projected gradient ascent for ease of implementation and convergence speed.  

%For each value of $n\in\{2,4,8,16,32,64\}$, we ran the experiment 100 times while generating agents' values for goods uniformly at random from the set $S=\{2^{2^{k-1}}|k\in[10]\}$ (i.e. for any agent~$i\in[n]$ and any good~$j\in[m]$, $v_{i,j}$ is picked uniformly at random from the set S). Generating agents' values uniformly from $S$ helps avoid convergence issues while running projected gradient ascent.

Our empirical results appear in Table~\ref{table:experiments}. As established in Corollary~\ref{cor:fair-division-add-goods} and Corollary~\ref{cor:ef2-division-add-goods}, the above procedure always finds an allocation which is $\Propone$ and $\EFtwo$. In fact, for about 96\% of the (randomly generated) instances, the implemented method finds an envy-free allocation. This suggests that, in practice, our algorithms outperform our theoretical guarantees. In addition, we find that it takes notably less time to execute the rounding method than to compute a market equilibrium (i.e., solve the Eisenberg-Gale program). 

\begin{table*}[h] 
	\caption{Empirical Results}
	\label{table:experiments}
	\footnotesize
	\begin{tabular}{|l|l|l|l|l|l|l|}
		\hline
		\textbf{Number of agents (n)} & $n=2$ & $n=4$ & $n=8$ & $n=16$ & $n=32$ & $n=64$\\
		\textbf{Number of goods (m)} & $m=10$ & $m=20$ & $m=40$ & $m=80$ & $m=160$ & $m=320$\\
		\hline
		\hline
		\textbf{Mean run-time of Gradient Ascent} & 1.104 sec & 1.621 sec & 2.067 sec & 2.869 sec & 5.593 sec & 6.559 sec\\
		\textbf{Mean run-time of Algorithm~\ref{alg:rearrange-spending}} & 0.0007 sec & 0.005 sec & 0.020 sec & 0.067 sec & 0.198 sec & 1.033 sec\\
		\textbf{Mean run-time of $\Alg$} & 0.0002 sec & 0.0005 sec & 0.0007 sec & 0.002 sec & 0.007 sec & 0.025 sec\\
		\hline
		\hline
		\textbf{Max run-time of Gradient Ascent} & 1.747 sec & 3.897 sec & 4.155 sec & 10.006 sec & 29.11 sec & 7.788 sec\\
		\textbf{Max run-time of Algorithm~\ref{alg:rearrange-spending}} & 0.001 sec & 0.011 sec & 0.050 sec & 0.109 sec & 0.299 sec & 1.329 sec\\
		\textbf{Max run-time of $\Alg$} & 0.001 sec & 0.005 sec & 0.002 sec & 0.004 sec & 0.012 sec & 0.038 sec\\
		\hline
		\hline
		\textbf{Number of $\EF$ allocations (out of 100)} & 99 & 86 & 95 & 99 & 98 & 100\\
		\textbf{Number of $\EFone$ allocations (out of 100)} & 100 & 86 & 95 & 99 & 98 & 100\\
		\textbf{Number of $\EFtwo$ allocations (out of 100)} & 100 & 100 & 100 & 100 & 100 & 100\\
		\textbf{Number of $\Prop$ allocations (out of 100)} & 99 & 86 & 96 & 100 & 100 & 100\\
		\textbf{Number of $\Propone$ allocations (out of 100)} & 100 & 100 & 100 & 100 & 100 & 100\\
		\hline
	\end{tabular}
\end{table*}

\bibliographystyle{alpha}
\bibliography{Bibliography-File}

\appendix

\section{Proof of Claim~\ref{claim:spending-forest}}
\label{app:spending-forest}

\renewcommand{\floatpagefraction}{.8}%Ensures that algorithm is the only float on a page, with no other text
%\RestyleAlgo{boxruled}%Draws a box around the algo environment
\begin{algorithm}%[H]
	\DontPrintSemicolon
	%\SetAlgoNoLine
	%\SetKwInOut{Input}{input}\SetKwInOut{Output}{output}
	\KwIn{A Fisher market $\M$ and its equilibrium $(\x,\p)$.}
	\KwOut{An equilibrium $(\x',\p)$ of $\M$ with the property that $G(\x',\p)$ is a forest.}
	% \Parameter{$0 < \varepsilon < 1$.}
	\BlankLine
	$\x'\leftarrow\x$.\;
	\While{there is a cycle in $G(\x',\p)$.}{
		Let $G\leftarrow G(\x',\p)$ and for any edge $(i,j)$ let $w_{i,j} :=x'_{i,j} p_j$ denote its weight in $G$.\;
		Find a cycle $C$ in $G$.\;
		Find a least weight edge on the cycle $C$ and let $w$ be the weight of this edge.\;
		\tcc{i.e., we pick an edge from the set $\arg\min_{(i,j)\in C}x'_{i,j}p_j$.}
		In the graph $G$, alternatingly add and subtract the weight $w$ from the edges of the cycle $C$ so that the least weight edge gets deleted.\;
		For all $(i,j)\in C$, $x'_{i,j}\leftarrow w_{i,j}/p_j$.
	}
	%	\While{$G(\x',\p)$ has a cycle.}{
	%	Find a cycle $i_1,j_1,i_2,j_2,\dots,i_k,j_k,i_1$\;
	%	}
	%\BlankLine
	\BlankLine
	\BlankLine
	\caption{Procedure to rearrange spending so that spending graph is a forest.}
	\label{alg:rearrange-spending}
\end{algorithm}

In this section we will show that Algorithm~\ref{alg:rearrange-spending} finds, in  strongly polynomial time, an allocation $\x'$ that satisfies Claim~\ref{claim:spending-forest}. Algorithm~\ref{alg:rearrange-spending} initializes $\x'$ to be the input allocation $\x$ and keeps iteratively modifying $\x'$. In every iteration of the while-loop, an edge which was part of a cycle $C$ of $G(\x',\p)$ gets deleted and, hence, every iteration deletes a cycle from the spending graph. Throughout these modifications, we maintain the invariant that $(\x',\p)$ is an equilibrium of the given Fisher market $\M$. These observations establish the stated claim and are detailed below. \\

\noindent
\textbf{Proof of Correctness:} Since in every iteration the graph $G = G(\x',\p) $ considered by the algorithm is bipartite, the selected cycles are always of even length. Hence, as we alternately add and subtract the least weight though a cycle, no agents total spending ever changes. Specifically, consider a cycle $i_1j_1i_2j_2\dots i_k j_k i_1$ in $G$. Without loss of generality, we can assume that $(i_1,j_1)$ is the least weight edge in the spending graph and let $w$ be the weight associated with this edge. We will delete the edge $(i_1,j_1)$ by subtracting the weight $w$ from it and alternately add and subtract $w$ throughout the cycle. Hence for any agent (say agent~$i_\ell$), we increase $i_\ell$'s spending on $j_{\ell-1}$ by $w$ and decrease $i_\ell$'s spending on $j_{(\ell+1) \mod k))}$ by $w$. Hence the total amount spent by any agent does not change. Moreover, as $w$ is the least weight of any edge in the cycle, adding or subtracting any agents spending on any good in the cycle by $w$ will maintain the non-negativity of all the spendings. Therefore, throughout the execution of the algorithm, the budget-exhaustion condition is maintained. 

For any good $j_\ell$ on the cycle, the consumption by agent $i_\ell$ decreases by $w/p_{j_\ell}$ and the consumption of this good by the agent $i_{(\ell+1) \mod k))}$ goes up by $w/p_{j_\ell}$. Therefore, the market-clearing conditions are maintained as well. Finally, the maximum bang-per-buck condition is also maintained. This follows from the fact that in $\x$ agents only spent on goods which provided them maximum bang-per-buck and in $\x'$ agents only spend on a subset of goods that they originally spent on in $\x$. 

These observations imply that Algorithm~\ref{alg:rearrange-spending} maintains the invariant that $(\x',\p)$ is an equilibrium of $\M$. From the condition in the while-loop, it is clear that the algorithm terminates if and only if $G(\x',\p)$ is a forest. We will now complete the proof of the claim by showing that Algorithm~\ref{alg:rearrange-spending} terminates in $\poly(n,m)$ time.\\

\noindent
\textbf{Run-Time Analysis:} In each iteration of the algorithm we delete one edge from $G(\x',\p)$ and never add a new edge to the graph. Therefore, the algorithm iterates at most $nm$ times. Furthermore, each iteration runs in strongly polynomial time, since it entails finding a cycle and a minimum weight edge on it. Therefore,  Algorithm~\ref{alg:rearrange-spending} terminates in strongly polynomial time.
This completes the proof. 

\section{Hardness of Finding Integral Equilibria}
\label{appendix:hardness}

\begin{theorem}
	\label{theorem:pure-hardness}
	It is {\rm NP}-hard to determine whether a given Fisher market admits an integral equilibrium or not.    
\end{theorem}
\begin{proof}
	
	We establish the hardness of determining whether a market is pure by reducing the partition problem to it. Recall that in the partition problem, we are given a set $S=\{s_1,s_2,\ldots, s_m\}$ of positive integers and the goal is to find a $2$-partition $(S_1,S_2)$ of $S$ such that the sum of the numbers in $S_1$ is equal to the sum of number in $S_2$. Given an instance of the partition problem with $m$ positive integers, we will construct a market with two agents and $m$ goods, $\M=\langle [2], [m], \V, \e\rangle$. Here, both agents have equal budget, $e_1=e_2=\frac{1}{2}\sum_{s\in S}s$, and identical, additive valuation $v_{1,j}=v_{2,j}=s_j$ for all $j\in[m]$.
	
	Note that $(\x,\p)$ is an integral equilibrium of $\M$ iff:
	\begin{itemize}
		\item Market clearing: $\x=(\x_1,\x_2)$ is a partition of $[m]$; in particular, for all $j\in[m]$ we have  $j \in \x_1\cup\x_2$.
		\item Maximum bang-per-buck allocation: for both the agents  $i \in \{1,2\}$ and for each good $j \in \x_i$ the $\MBB$ condition implies that  $\frac{v_{i,j}}{p_j}=\max_{j'\in[m]}\frac{v_{i,j'}}{p_j'}$. Since both agents have the same valuation, we have $\frac{s_j}{p_j}=\frac{s_{j'}}{p_{j'}}$ for all $j.j' \in [m]$.
		\item Budgets exhaustion: $\sum_{j\in\x_i}p_j=e_i$ for all $i\in\{1,2\}$. Therefore, $\sum_j p_j = e_1+e_2 = \sum_{s\in S}s$. This along with the fact that $\frac{s_j}{p_j}=\frac{s_{j'}}{p_{j'}}$ for all $j,j'\in[m]$ implies that $\p=(s_1,s_2,\dots,s_m)$.
	\end{itemize}
	This implies that $(\x,\p)$ is a market outcome of $\M$ iff $\p=(s_1,s_2,\dots,s_n)$  and \begin{align*} \sum_{j\in \x_1}s_j=\sum_{j\in \x_1}p_j=\e_1=\e_2=\sum_{j\in \x_1}p_j=\sum_{j\in \x_2}s_j.\end{align*}
	
	Hence, there exists a integral equilibrium for $\M$ iff there exists a $2$-partition $(S_1,S_2)$ of $S$ such that the sum of the numbers in $S_1$ is equal to the sum of the numbers in $S_2$. This establishes the stated claim. 
\end{proof}

\section{Comparative Example}
\label{appendix:example}

This section provides an example of a Fisher market wherein $\Alg$ outperforms (in terms of budget perturbations) the algorithm developed in~\cite{barman2018finding}.

Consider a market that consists of $4n-1$ goods and $2n$ agents, each with a budget of $1$. The first $n$ agents value the first $2n$ goods uniformly at $n$. In addition, the first $n$ agents have a value of zero for the last $2n-1$ goods. The last $n$ agents value the first $2n$ goods uniformly at $(1-\varepsilon)$ and their value for each of the last $2n-1$ goods is equal to one. 

At equilibrium, each of the first $2n$ goods will be priced at $1/2$. In addition, the equilibrium prices of last $2n-1$ goods will be $n/(2n-1)$, each.

Therefore, via Theorem~\ref{thm:strong-poly-mark}, we can find a pure market by perturbing the budgets no more than $n/(2n-1) \approx 1/2$. Next, we will show that the pure market obtained via the algorithm in~\cite{barman2018finding} leads to a budget perturbation of $\approx 3/4$.

The algorithm of Barman et al.~\cite{barman2018finding} would start with a welfare-maximizing allocation, i.e., it would start by allocating (i) the first $2n$ goods among the first $n$ agents and (ii) the last $2n-1$ goods among the last $n$ agents. Note that, under this allocation, one of  the last $n$ agents gets less than two goods. In~\cite{barman2018finding} the prices are initialized to be equal to the valuations; one can normalize them after the termination of the algorithm to ensure that the sum of prices is equal to the sum of the budgets, i.e.,  equal to $n$. 

Since this initial allocation is not \emph{price envy-free up to one good}, the algorithm of Barman et al.~\cite{barman2018finding} would scale the prices up and, in particular, increase the prices of the last $2n-1$ goods to $n$ each.  At this point of time, price envy-freeness up to one good is achieved and the algorithm would terminate. Overall, the method in~\cite{barman2018finding} will find a solution in which every good is priced at $n$ and there exists an agent who receives exactly one good; the remaining agents will obtain two goods, each.

Note that, at this point, however, the sum of prices is $T := 2n \cdot n + (2n-1) \cdot n = 4n^2 - n$. To get the sum of prices back to $n$ we scale them down by $T/n$. Hence, the budget of the agent with a single good scales down to $n \ \frac{n}{T} = \frac{n^2}{4n^2 - n} \approx \frac{1}{4}$. Therefore, the change in the budget of this agent is about $ 1 - 1/4 = 3/4$. As mentioned previously, the algorithm developed in the present paper would lead to a budget perturbation of close to $1/2$ and, hence, will perform better on this instance. 

\end{document}